\documentclass[a4paper, 12 pt, article]{ieeeconf}      % Use this line for a4
\IEEEoverridecommandlockouts                              % This command is only
\usepackage{mathptmx} % assumes new font selection scheme installed
\usepackage{times} % assumes new font selection scheme installed
\usepackage{amsmath} % assumes amsmath package installed
\usepackage{amssymb}  % assumes amsmath package installed

\title{\LARGE \bf
Solving the Fermi paradox without assumptions*
}
\author{Berezin Alexander% <-this % stops a space
\thanks{*This work was not supported by any organization.}% <-this % stops a space
}

\newtheorem{definition}{Definition}[section]
\newtheorem{lemma}[definition]{Proposition}

\begin{document}
\maketitle
\thispagestyle{empty}
\pagestyle{empty}

%%%%%%%%%%%%%%%%%%%%%%%%%%%%%%%%%%%%%%%%%%%%%%%%%%%%%%%%%%%%%%%%%%%%%%%%%%%%%%%%
\begin{abstract}
This paper suggests that a universal solution to the Fermi paradox exists and can be derived directly from the definition of life and/or intelligence, therefore eliminating the need for any questionable assumptions and even for the anthropic principle. The proposed solution\footnote{A more colloquial name for this solution emerged during its discussion: ``the \textbf{Black attractor}''.} puts an upper limit on growth of civilizations that is independent of resource availability or biological factors.
\end{abstract}

%%%%%%%%%%%%%%%%%%%%%%%%%%%%%%%%%%%%%%%%%%%%%%%%%%%%%%%%%%%%%%%%%%%%%%%%%%%%%%%%
\section{Introduction}

It is commonly understood that infinite growth is physically impossible on a finite planet \cite{meadows_limits_2004}. However, at first glance, the known laws of physics suggest no hard limits to growth for civilizations that have the capability of interstellar travel, barring the speed of light and the total amount of accessible energy in the Universe. This is where the Fermi paradox appears.

Given even the most conservative estimations of growth rate, a civilization moderately older than humanity should have already left easily detectable traces, or even populated the entire galaxy \cite{Freitas_1980}. The lack of such observations is usually explained by suggesting that either the odds of a capable enough civilization to appear are extremely low \cite{ward_2000_rare}, or that they necessarily destroy themselves or each other \cite{Brin_1983}.

This paper argues that there is, indeed, a universal internal limit to growth of civilizations, even without any external limiting factors. Another way of putting this conclusion is to say that any technology that permits any civilization to leave a noticeable signature on a stellar scale inherently presents an existential risk in the sense that it can (and inevitably will) be used to permanently and drastically curtail the potential of said civilization \cite{Bostrom2002-BOSERA}.

\section{Model}
In order to discuss the Fermi paradox, at least one definition has to be agreed upon: usually it is either life, intelligence or civilization. However, making distinctions between these concepts may in itself be an unwarranted assumption. In principle, it is possible that some forms of life could evolve the capability of interplanetary, or even interstellar travel without evolving civilization, or even intelligence, first. The distinction between unintelligent and intelligent organisms does not have to be discrete. Even the concept of an organism does not have to retain its meaning for alien life.

In order to discourage the reader from making assumptions based on colloquial definitions of these terms, this paper will be referring to \textit{``agents''} as a generalization of all entities to which the Fermi paradox is applicable.

At least three attempts at such a generalization have been published. One is Erwin Schrodinger's ``What Is Life?'' \cite{schrodinger_what_1992}, the second is ``Equation of intelligence'' by Alexander Wissner-Gross \cite{wissner-gross_causal_2013}, and the third is ``Free energy principle'' by Karl Friston \cite{Friston20130475}.

\begin{definition}
\label{Life}
\cite{schrodinger_what_1992}
A system is considered \textit{alive} if it can decrease or maintain its internal entropy by increasing the entropy of its surroundings.
\end{definition}

\begin{definition}
\label{Intelligence}
\cite{wissner-gross_causal_2013}
A system is considered \textit{intelligent} if its actions aim at maximizing its future freedom of action.
\end{definition}

\begin{definition}
\label{Free energy}
\cite{Friston20130475}
A living system aims to minimize the dispersion of its sensed states, while using those sensations to infer external states of the world.
\end{definition}

\begin{lemma}
\label{equivalence}
The three definitions above are equivalent at scales relevant to the Fermi paradox.
\end{lemma}
\begin{proof}
First, let us consider a system that aims to maximize its future freedom of action. As described in \cite{wissner-gross_causal_2013}, when presented with a set of possible action paths, such system will choose the path that maximizes an average of short-term directions weighted by the diversity of long-term paths that they make reachable. However, precise integration over long-term paths requires full information about all factors that can possibly affect the system, i.e. zero dispersion of sensed states. As the dispersion increases, the precision of path integration can only decrease. Therefore, \ref{Free energy} follows from \ref{Intelligence}.

Second, maintaining or decreasing the internal entropy of a system requires continuous increase in its external entropy. Hence, the system either ceases being alive once its energy reservoir is exhausted or seeks more energy in the external environment. And since any action increases entropy, seeking more energy over long enough time is equivalent to maximizing future freedom of action. Therefore, \ref{Intelligence} follows from \ref{Life}.

Finally, as noted by Friston himself in \cite{Friston20130475}, ``Under ergodic assumptions, the long-term average of surprise (a.k.a. ``free energy'') is entropy'', where ``ergodicity means that the time average of any measurable function of the system converges (almost surely) over a sufficient amount of time''. Considering the timescales relevant to the Fermi paradox, we may conclude that \ref{Life} follows from \ref{Free energy}.
\end{proof}

Let \textit{agent} be an entity that conforms to any one of the equivalent definitions; and let $\mathbb{A}$ be the set of all possible agents. A \textit{civilization} is then defined as an arbitrary nonempty subset of $\mathbb{A}$.

\section{Solution}
\label{solution}

\begin{lemma}
\label{greed}
Maximizing future freedom of action is equivalent to hoarding the greatest amount of resources.
\end{lemma}
\begin{proof}
Any action can be described as a combination of movements. Movements from lower to higher energy states\footnote{The ``energy state'' here is used in the general sense and does not imply the states are discrete.} expend energy. Any theoretically possible process of energy generation expends some material resource as fuel. Therefore, any action expends some fuel material, and maximizing the amount of fuel also maximizes the space of available actions.
\end{proof}

Obviously, materials have more uses than just fuel; they are required to sustain any type of growth. However, if the most efficient reactor possible converts any matter into energy with no waste products and any fuel can be converted into some construction material, then the terms \textit{``fuel''} and \textit{``matter''} can be used interchangeably. Consequently, as the efficiency of technology approaches its theoretical maximum, correlation between the \textit{value} of resources and their physical mass becomes more linear, meaning that value is also interchangeable with mass.

At this point, a few more definitions should be added for readability purposes.
\begin{definition}
An agent's \textit{weight} is the cumulative value of resources that agent controls.
\end{definition}
\begin{definition}
An \textit{$\alpha$-agent} is the one that has the greatest weight in a civilization.
\end{definition}

\begin{lemma}
\label{pareto}
Given long enough time, all available value will be controlled by the $\alpha$-agent.
\end{lemma}
\begin{proof}
Consider $N$ agents each starting with $weight(n) = s - n*\epsilon$, where $s > 0$ is an arbitrary constant and $\epsilon$ is a minor perturbation: $0 < \epsilon << s$. Assume they invest all weight into obtaining more value from space and their technological solutions are equally efficient.

Reaching resources in space requires to change one's energy state. Since energy is equivalent to value, the $\alpha$-agent (with corresponding $n=0$) will be able to perform the greatest change of energy state, and therefore gain access to the greatest amount of new value. The $\alpha$-agent now has a weight proportional to $s^2$, while other agents have $weight(n) \propto (s-n\epsilon)^2$. Repeating the same process will result in agent $n$ on step $t$ having $weight(n) \propto (s-n\epsilon)^t$.

$$\lim\limits_{t \to \infty} \frac{s^t}{\sum_{n=0}^{N} (s-n\epsilon)^t} = 1$$

As $t \to \infty$, the fraction of value held by the $\alpha$-agent approaches 1, even assuming completely fair competition. This means that all the available value will at some point belong to the $\alpha$-agent.
\end{proof}

\begin{lemma}
\label{storage}
If the number of agents in a civilization is greater than one, they will be incentivized to centralize storage of their value.
\end{lemma}
\begin{proof}
Since the definition of agents \ref{Intelligence} does not forbid them to steal from each other, such behaviour is preferable when the value obtained surpasses the value lost in the conflict (according to \ref{greed}). Knowing this, agents will attempt to maximize the value that any other agent will lose by attacking them while minimizing their own expenditures on that task. Considering that the value spent on defence of a region of space is proportional to its surface area, while the value that can be stored inside that region is proportional to its volume, it now follows from basic geometry that the optimal way to store value is within a single spherical region.
\end{proof}

Whether a civilization consisting of a single agent can evolve to a scale significant to the Fermi paradox is unclear. This paper does not cover such a possibility.

\begin{lemma}
\label{black}
For every civilization of more than one agent there exists an upper bound to the amount of resources it can accumulate. Attempted growth past that bound will result in a collapse into a black hole.
\end{lemma}
\begin{proof}
According to \ref{pareto}, a disproportionate amount of resources is controlled by a single agent. According to \ref{storage}, all these resources are packed very closely. From here the problem is obvious. A huge mass packed densely enough produces a black hole.
\end{proof}

\section{Testability}
Usually testing a Fermi paradox solution proposal would require observations of alien life (or lack thereof), which makes the proposals essentially untestable at present.

However, if \ref{black} is true and extraterrestrial civilizations had existed at some point in our galaxy, we should be able to observe black holes left by their collapse. Problem is, attributes of those black holes depend on other unknown quantities, so telling whether a specific body is a remnant of a past civilization or a natural formation depends entirely on context.

Basically, observing multiple black holes that do not fit accepted astronomical models of star/galaxy formation might be considered as weak evidence for this hypothesis, and vice versa. Although observing black holes without prior knowledge of their location is extremely difficult, it is suggested in \cite{abbott_observation_2016} that gravitational astronomy may change that in near future.

\section{Discussion}
\label{Discussion}

\subsection{How does this explain the absence of von Neumann probes?}

\begin{definition}
\label{vonNeumann}
A von Neumann probe (VNP) is an entity that aims to place at least one copy of itself at the greatest number of stars in the shortest time period.
\end{definition}
\begin{lemma}
VNPs are agents.
\end{lemma}
\begin{proof}
It is fairly obvious that reducing the time of interstellar flights requires increasing the amount of resources spent by probes. Considering that resource extraction and replication are relatively fast processes compared to interstellar flights, VNPs are incentivized to gather as much resources as possible before departing for the next star. According to \ref{greed}, this proves their agency.
\end{proof}

Therefore, VNPs face the same problem as civilizations, if they are not the same thing to begin with. Whether VNPs are manned or autonomous is irrelevant as long as they inherit the \ref{Intelligence} property.

\subsection{How does one even move such amounts of mass for such distances in order to create a black hole?}

We might not know yet what technologies are going to be used for such large-scale projects; what we can say definitively, though, is that the laws of physics as we understand them are not the limitation.

\subsection{But if the collapse is so easy to predict, and there is a sentient entity in control, that entity will not allow it.}

The key point is that no agent is actually in control. Rather, the established system is in control of the agents. A more intuitive example of such a situation is the Prisoners' dilemma \cite{poundstone_prisoners_1993}. The same description can be applied to governments accumulating advanced weapons that can never be used, bankers' behaviour preceding a financial crash and the global response to climate change. Generally speaking, there is no basis to claim that systems consisting of rational agents also act rationally.

\subsection{But if the collapse is so easy to predict, surely most individuals should be able to escape it?}

Problem is, escaping the collapse requires escaping the gravity of the collapsing structure, which would require exponentially more energy as its mass increases; at the same time, the total energy stored within it only grows linearly, and the vast majority of individuals inhabiting it will not possess any meaningful fraction of it according to \ref{pareto}.

On the other hand, some individuals might me able to escape by sheer chance, and some will not be caught in the collapsing region at all. However, all wealth the civilization had accumulated is irreversibly destroyed. Survivors have to start from scratch. And when they finally manage to recreate what was lost, history will simply repeat itself.

\subsection{In a civilization consisting of billions or trillions of individual agents, how can one of them exercise so much power without the help of others? If the agents on whom that one relies organize and stage a revolt, they should be able to divide the resources and avoid a collapse, shouldn't they?}
\label{revolution}

\begin{lemma}
The probability of successful revolt against the $\alpha$-agent decreases as the size of civilization (in terms of cumulative value) increases.
\end{lemma}
\begin{proof}
An ideal revolt is a conflict between the $\alpha$-agent and the sum of all other agents in a given civilization. Assuming the defensive and offensive systems do not need their own agency to function, the outcome of a conflict is only determined by the value that each side is able to spend on it. According to \ref{pareto}, the fraction of value held by the $\alpha$-agent increases as civilization grows. Eventually its weight will surpass that of all other agents put together, at which point a successful revolt is impossible.
\end{proof}

It is, indeed, possible until then, but, once the growth resumes, the same dynamic from \ref{pareto} also continues, eventually returning the value distribution to the same shape.

\subsection{If agents are meant to maximize their future freedom of action (\ref{Intelligence}), why wouldn't they factor in the probability of a collapse?}

Even if they do, this should not affect their behavior. If they choose to suspend growth for safety, another agent will exploit that, as explained in \ref{storage}. The risk will grow with time no matter what they do. And, given the time scale that an interstellar civilization needs for development, that risk accumulates to a near certainty of collapse.

\subsection{What if the civilization is forced to constantly spend the resources they acquire?}
\label{spending}

That would protect the civilization itself from a collapse. However, the only feasible condition that ensures such spending is a war with another civilization. Whether this is a better alternative to collapse is up to the reader to decide.

\subsection{Can the black holes resulting from collapsed civilizations explain dark matter?}

No. The idea of dark matter consisting of black holes is currently not considered viable \cite{zumalacarregui_limits_2018}.

\section{Implications}

\subsection{Assuming this solution is correct, what predictions can it make concerning the future of humanity?}

In the near future, private space industry should grow dramatically. As already predicted in \cite{lewis_1996_mining}, just the mining of near-Earth asteroids can bring unprecedented profit in minerals that are rare in Earth's crust, but abundant in smaller celestial bodies. That capital can then be invested into even more ambitious projects requiring permanent settlements in outer space, such as building solar reflectors in Earth's L1 point to decrease temperature and mitigate climate change; setting up high-powered lasers to propel interstellar spacecraft \cite{lubin_2016_roadmap} or even curing aging as an indirect result of modifying humans for long-term exposure to microgravity and radiation \cite{biolo_2003_microgravity}. Slowly but inevitably, engineering challenges that seem insurmountable today will become real targets; challenges such as gathering enough resources to collapse into a black hole.

\subsection{How can a civilization prevent its collapse?}

Two ways have been suggested in \ref{revolution} and \ref{spending}. The third solution might come from the fact that it is theoretically possible for few civilizations to overcome their economic incentives. Even if they themselves survive, other uncontrollably expanding civilizations would still endanger their existence. That should incentivize forcibly constraining younger civilizations, implementing ``the Zoo hypothesis'' \cite{BALL1973347} in practice. Such scenario does not permit unconstrained expansion either, thus remaining an equally viable solution to the Fermi paradox.

%\vfill\null
\bibliographystyle{unsrt}
\bibliography{black_attractor.bib}
\end{document}